\documentclass{llncs}
\usepackage{amsmath}
\usepackage{graphicx}
\usepackage{amssymb}
\usepackage{verbatim}
\usepackage{epsfig}
\usepackage{color}
\usepackage{url}
\usepackage{listings}
\usepackage{float}
\usepackage[caption=false]{subfig}

\usepackage{cite} 
\usepackage{authblk} 
\usepackage{appendix}
\usepackage{enumitem}
\usepackage{mathtools}

\spnewtheorem{Example}{Example}{\bfseries}{\rmfamily}
\spnewtheorem{observation}{Observation}{\bfseries}{\itshape}
\spnewtheorem{Observation}[observation]{Observation}{\bfseries}{\rmfamily}
\spnewtheorem{Algorithm}{Algorithm}{\bfseries}{\rmfamily}


\newcommand{\entry}{\mbox{\scriptsize \it entry}}
\newcommand{\exit}{\mbox{\scriptsize \it exit}}
\newcommand{\inside}{\mbox{\it inside}}
\newcommand{\outside}{\mbox{\it outside}}
\newcommand{\belongs}{\mbox{\it belongs}}
\newcommand{\dist}{\mbox{\it dist}}
\newcommand{\indeg}{\mbox{\it indeg}}

\renewcommand{\comment}[1]{}

\title{DAG-width of Control Flow Graphs with Applications to Model Checking\thanks{Research supported by NSERC}}

\author{Therese Biedl, Sebastian Fischmeister and Neeraj Kumar\thanks{Corresponding author}}
\institute{University of Waterloo, Canada \\ \email{\{biedl,sfischme,n26kumar\}@uwaterloo.ca}}
\date{}

\begin{document}
\maketitle
\begin{abstract}
The treewidth of control flow graphs arising from structured programs is known to be at most six.
However, as a control flow graph is inherently directed, it makes sense to consider a measure of
width for digraphs instead. We use the so-called DAG-width and show that the DAG-width of 
control flow graphs arising from structured (goto-free) programs is at most three. Additionally, we 
also give a linear time algorithm to compute the DAG decomposition of these control flow graphs.
One consequence of this result is that parity games (and hence the $\mu$-calculus model checking problem),
which are known to be tractable on graphs of bounded DAG-width, can be solved efficiently in practice 
on control flow graphs.
\end{abstract}

\section{Introduction}
Given a program $P$ and a property $\pi$, software verification concerns the problem of
determining whether or not $\pi$ is satisfied in all possible executions of $P$. The problem
can be formulated as an instance of the more general $\mu$-calculus model checking problem. Under
such a formulation, we specify the property $\pi$ as a $\mu$-calculus formula and evaluate it on
a model of the system under verification. The system is usually modeled as a state machine, like a
\emph{kripke structure}~\cite{kripke1963semantical}, where states are labeled with appropriate propositions
of the property $\pi$. For example, in the context of software systems, the \emph{control-flow graph} 
of the program is the kripke structure. The states represent the basic blocks in the program and 
the transitions represent the flow of control between them.

The complexity of the $\mu$-calculus model checking problem is still unresolved: the problem is
known to be in NP $\cap$ co-NP~\cite{EmersonJS01}, but a polynomial time algorithm has not been found. It is also known
that the problem is equivalent to deciding a winner in \emph{parity games}, a two-player game played
on a directed graph. More precisely, given a model $M$ and a formula $\phi$, we can construct a directed
graph $G$ on which the parity game is played~\cite{Stirling:2001}. Based on the winner of the game on $G$, we can determine
whether or not the formula $\phi$ satisfies the model $M$. Motivated by this development, parity games 
were extensively studied and efficient algorithms were found for special graph classes~\cite{obdrvzalek2003fast, berwanger2012dag, GanianHKLOR14}.

For software model checking, the result for graphs of bounded \emph{treewidth}~\cite{obdrvzalek2003fast} 
is of particular interest since the treewidth of control-flow graphs for structured (goto-free) programs
is at most 6 and this is tight~\cite{thorup1998all}. 
The proof comes with an associated tree decomposition 
(which is otherwise hard to find~\cite{ArnborgEtAl1987}).

Obdr\v{z}\'{a}lek \cite{obdrvzalek2003fast}
gave an algorithm for parity games on graphs of treewidth
at most $k$. The algorithm runs in $O(n \cdot k \cdot d^{2(k+1)^2})$ time, where $n$ is the number of vertices and
$d$ is the number of \emph{priorities} in the game. For the $\mu$-calculus model checking problem, the number
of priorities $d$ is equal to two plus the \emph{alternation depth} of the formula, which in turn is at most $m$ (the size of the formula).

In practice 
both $m$ and $d$ are usually quite small.
Since the treewidth of control-flow graphs is also small, Obdr\v{z}\'{a}lek believed that his result should give 
better algorithms for software model checking. However, as pointed out in~\cite{fearnley2011time}, the algorithm is far
from practical due to the large factor of $d^{2(k+1)^2}$. 
For example, 
the parity game (and hence model checking with a single sub-formula) on a control-flow graph of treewidth 6 will have a running 
time of $O(n \cdot d^{98})$. Fearnley and Schewe~\cite{fearnley2011time} have improved the run time for bounded treewidth graphs to 
$O(n \cdot (k+1)^{k+5} \cdot (d+1)^{3k+5})$. This brings the run time to $O(n \cdot 7^{11} \cdot (d+1)^{23})$, which still seems 
impractical. In the same paper, they also present an improved result for graphs of \emph{DAG-width} at most $k$, running 
in $O(n \cdot M \cdot k^{k+2} \cdot (d+1)^{3k+2})$ time. Here, $M$ is the number of edges in the DAG decomposition; no bound better than
$M \in O(n^{k+2})$ is known.

\paragraph{Contribution}
We observe that since treewidth is a measure for undirected graphs, explaining the structure of control-flow
graphs via treewidth is overly pessimistic, and by ignoring the directional properties of edges we may be losing
possibly helpful information. Moreover, software model checking could benefit from 
DAG-width based algorithms from~\cite{fearnley2011time,
DBLP:conf/csl/BojanczykDK14}, if we can find a DAG decomposition of control-flow graphs
with a small width and fewer edges.

To this end, we show that the DAG-width of control-flow graphs arising from structured (goto-free) programs 
is at most 3. Moreover, we also give a linear time algorithm to find the 
associated DAG decomposition with a linear number of edges.
Combining this with the results by Fearnley and Schewe~\cite{fearnley2011time}, parity games on control-flow graphs 
(and hence model checking with a single sub-formula) can be solved in $O(n^2 \cdot 3^{5} \cdot (d+1)^{11})$ time.
This is competitive with and probably more practical than the previous algorithms.

From a graph-theoretic perspective, it is desirable for a digraph width measure (see~\cite{GanianHKLOR14} for a brief survey)
to be small on many interesting instances. The above result makes a case for DAG-width by demonstrating that there are application
areas that benefit from efficient DAG-width based algorithms.

\paragraph{Outline} The remainder of the paper is organized as follows. Section~\ref{sec:prelims} reviews notation and 
background material, especially the \emph{cops and robbers game} and its relation to DAG-width.
The proof of the DAG-width bound follows in Section~\ref{sec:proof}. In Section~\ref{sec:decompose}, we discuss the algorithm to find
the associated DAG decomposition.
Finally, we conclude with Section~\ref{sec:conclusion}.

\section{Preliminaries}
\label{sec:prelims}


\subsection{Control-Flow Graphs} 

\begin{definition}
	The \emph{control-flow graph} of a program $P$ is a directed graph $G=(V, E)$, where a vertex $v \in V$ represents
	a statement $s_v \in P$ and an edge $(u, v) \in E$ represents the flow of control from $s_u$ to $s_v$ under some 
	execution of the program $P$. The vertices \textnormal{\texttt{start}} and \textnormal{\texttt{stop}} correspond
	to the first and last statements of the program, respectively.
\end{definition}

For the sake of having fewer vertices in $G$, most representations of control-flow graphs combine sequence of statements
without any branching into a \emph{basic block}. This is equivalent to contracting every edge $(u, v) \in E$, such that
the \emph{in-degree} and \emph{out-degree} of $v$ is $1$. 
Throughout the paper, we assume that the control-flow graphs are derived from structured (goto-free) programs.
This is done because with unrestricted gotos any digraph could 
be a control-flow graph and so they have no special 
width properties. We can construct a control-flow graph of a structured program by parsing the 
program top-down and expanding it recursively depending on the control structures as shown in Figure~\ref{fig:controlStructures}
(see~\cite{Aho1986} for more details). 
Following the same convention as~\cite{thorup1998all}, we note that the potential successors of a statement $S$ in the program $P$ are:
\begin{itemize}
	\item \emph{out}, the succeeding statement or construct
	\item \emph{exit}, the exit point the nearest surrounding loop (\texttt{break})
	\item \emph{entry}, the entry point of the nearest surrounding loop (\texttt{continue})
	\item \emph{stop}, the end of program (\texttt{return, exit})
\end{itemize}

\begin{figure}
	\centering
	\subfloat[\mbox{}]{\includegraphics[scale=0.54]{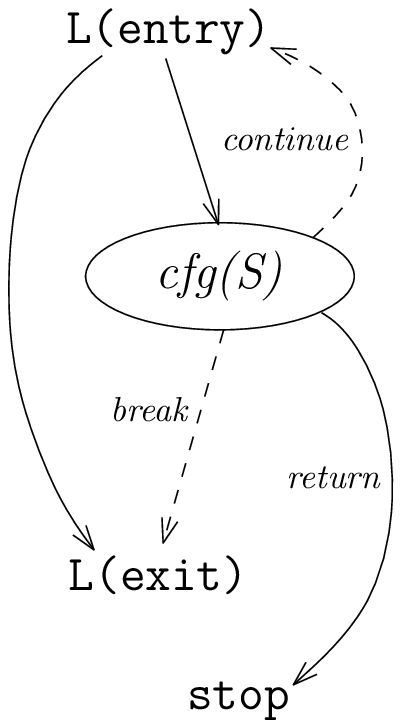}\label{fig:loopElement}} 
	\hfill
	\subfloat[\mbox{}]{\includegraphics[scale=0.51]{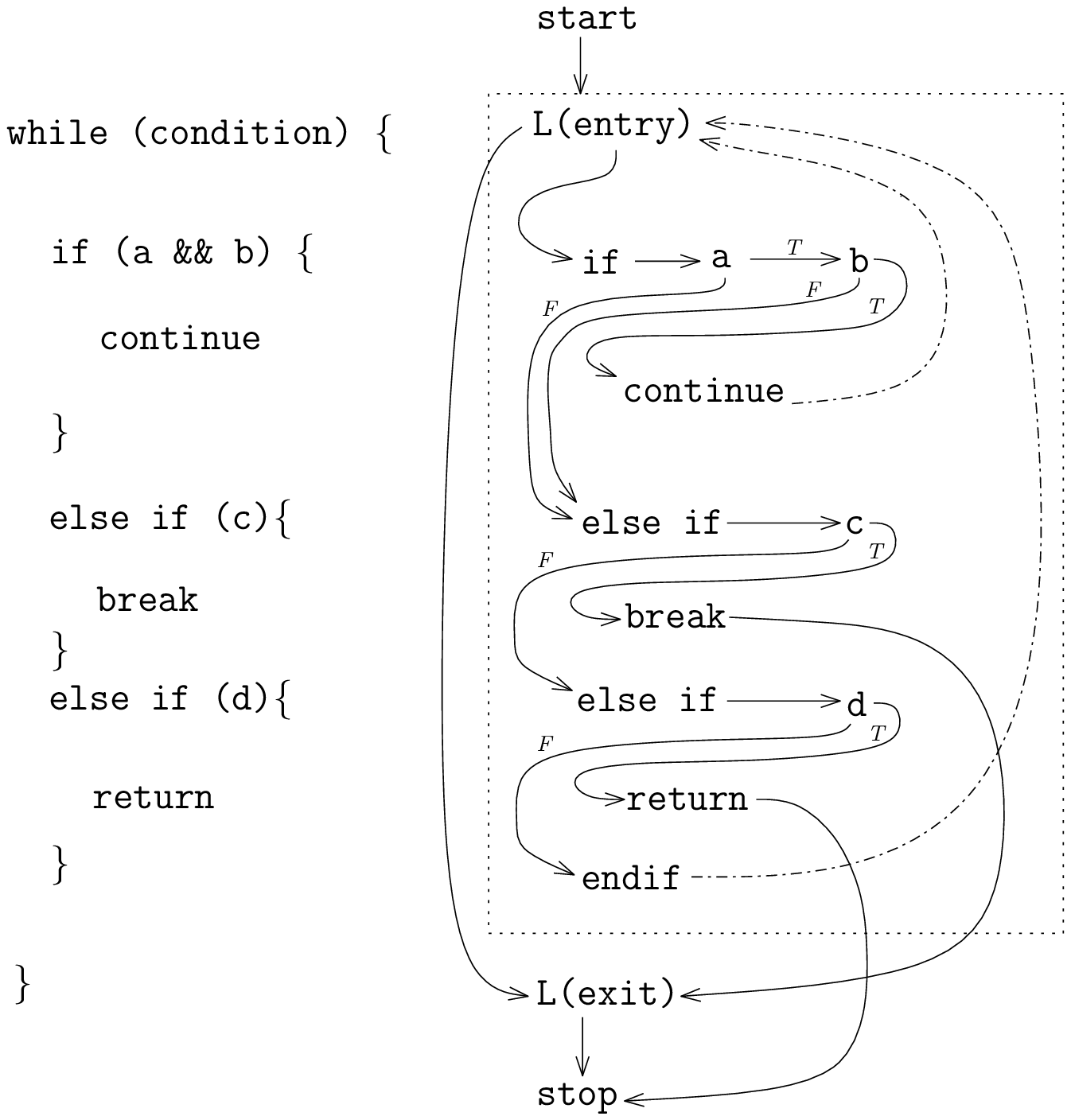}\label{fig:cfg-example} }
	\caption{Loop elements.  (a) The abstract structure.  
Dashed edges must start in $\belongs(L)$. (b) Sample-code and the resulting control-flow graph.
		Backward edges are dash-dotted. $\inside(L)$ is shown dotted.}
	\label{fig:controlStructures}
\end{figure}

\begin{definition} [Dominators] Let $G$ be a control-flow graph and $u, v \in V$. We say that $u$ \emph{dominates} $v$,
	if every directed path from \textnormal{\texttt{start}} to $v$ must go through $u$. Similarly, we say that
	$u$ \emph{post-dominates} $v$, if every directed path from $v$ to \textnormal{\texttt{stop}} must go through $u$.
\end{definition}

\begin{definition} [Loop Element]
	For every loop construct such as \textnormal{\texttt{do-while}}, \textnormal{\texttt{foreach}},
	we can construct an equivalent representation as a \emph{loop element} $L$, characterized
	by an entry point $L^{\entry}$ and an exit point $L^{\exit}$. See	Figure~\ref{fig:loopElement}.
\end{definition}

We note the following definitions and properties for loop elements:
\begin{itemize}
	\item We define $\inside(L)$ to be the set of vertices dominated by $L^{\entry}$
		and not dominated by $L^{\exit}$.
		Quite naturally, we define 
		$\outside(L)$ to be the set of vertices $(V \setminus \inside(L))$. 
		Note that $L^{\entry}\in \inside(L)$ 
		but $L^{\exit}\in \outside(L)$. Moreover, if we ignore edges to \texttt{stop},
		$L^{\exit}$ post-dominates vertices in $\inside(L)$.

	\item For the purpose of simplification, we assume $G$ to be enclosed in a hypothetical loop element $L_\phi$.
		This is purely notational and we do not add extra vertices or edges to $G$. We have
		$\inside(L_\phi) = V$ and $\outside(L_\phi) = \emptyset$.

	\item We say that a loop element $L_i$ is \emph{nested} under $L$, iff $L_i^{\entry} \in \inside(L)$. Two distinct loop
		elements are either nested or have disjoint insides.
	
	\item We can now associate every vertex of $G$ to a loop element as follows. We say that a vertex $v \in V$ \emph{belongs to}
		$L$ if and only if $L$ is the nearest loop element such that $L^{\entry}$ dominates $v$. More precisely, $v \in \belongs(L)$
		if and only if $v \in \inside(L)$, and there exists no $L_i$ nested under $L$ with $v \in \inside(L_i)$.
		
	\item Every $v \in V$ belongs to exactly one loop element. \texttt{start} and \texttt{stop} (as well as any vertices outside all loops of the program) belong to $L_\phi$.

	\item Finally, we say that a loop element $L_i$ is \emph{nested directly} under $L$, iff $L_i^{\exit} \in \belongs(L)$. In 
		other words, $L_i$ is nested under $L$ and there exists no $L_j$ nested under $L$ such that $L_i$ is nested under $L_j$.
\end{itemize}

We say that an edge $(u, v) \in E$ is a \emph{backward edge}, 
if $v$ dominates $u$; otherwise we call it a \emph{forward edge}. 
The following observations will be crucial:

\begin{lemma}\cite{allen1970control}
\label{lem:backward}
The backward edges are exactly those that lead from a vertex in $\belongs(L)$ to
$L^{\entry}$, for some loop element $L$.
\end{lemma}
\begin{corollary}
\label{cor:cycle}
Let $C$ be a directed cycle for which all vertices are in $\inside(L)$ and
at least one vertex is in $\belongs(L)$, for
some loop element $L$. Then $L^{\entry}\in C$.
\end{corollary}

\subsection{Treewidth and DAG-width}
The \emph{treewidth}, introduced in~\cite{SeymourRobertson1986}, is a graph theoretic concept which measures
tree-likeness of an undirected graph. We will not review the formal definition of treewidth here since we 
do not need it. Thorup~\cite{thorup1998all} showed that every control-flow graph has a treewidth of at most $6$.   This implies that any control-flow graph
has $O(n)$ edges.
%

The \emph{DAG-width}, introduced independently in~\cite{BDHK06, Obd06}, is a measure of how close a directed graph is to
being a directed acyclic graph (DAG). Like the treewidth, it is defined
via the best possible width of a so-called {\em DAG-decomposition}.
We will review the formal definition in Section~\ref{sec:decompose}.

\subsection{Cops and Robbers game}
The \emph{cops and robber} game on a graph $G$ is a two-player game in which $k$ cops attempt to catch a robber.
Most graph width measures have an equivalent characterization via a variant
of the cops and robber game.
For example, an undirected graph $G$
has \emph{treewidth} $k$ if and only if $k+1$ cops can search $G$ and successfully catch the robber~\cite{seymour1993graph}.

The DAG-width relates to the following variant of the cops and robber game
played on a directed graph $G=(V, E)$:
\begin{itemize}
	\item The cop player controls $k$ cops, which can occupy any of the $k$ vertices in the graph. We denote this set as
		$X$ where $X \in [V]^{\leq k}$. The robber player controls the robber which can occupy any vertex $r$.

	\item A play in the game is a (finite or infinite) sequence $(X_0, r_0), (X_1, r_1), \dots, (X_i, r_i)$ of positions taken
		by the cops and robbers.  $X_0=\emptyset$, i.e., the robber
		starts the game by choosing an initial position.

	\item In a transition in the play from $(X_i, r_i)$ to $(X_{i+1}, r_{i+1})$, the cop player 
		moves the cops not in $(X_i \cap X_{i+1})$ to $(X_{i+1} \setminus X_i)$ with a helicopter.
		The robber can, however, see the helicopter landing and move at a great speed 
along a cop-free path to
		another vertex $r_{i+1}$.
More specifically, there must be a directed path 
		from $r_i$ to $r_{i+1}$ in the digraph $G \setminus (X_i \cap X_{i+1})$.

	\item The play is winning for the cop player, if it terminates with $(X_m, r_m)$ such that $r_m \in X_m$. If the play is infinite,
		the robber player wins.

	\item A ($k$-cop) strategy is a function $f : [V]^{\leq k} \times V \rightarrow [V]^{\leq k}$. Put differently, the cops can see the robber when
deciding where to move to.
A play is consistent
		with strategy $f$ if $X_{i+1} = f(X_i, r)$ for all $i$.
\end{itemize}

\begin{definition}(Monotone strategies)
	\label{def:monotone}
	A strategy for the cop player is called \emph{cop-monotone}, if in a play consistent with
	that strategy, the cops never visit a vertex again. More precisely, if $v \in X_i$ and $v \in X_k$
	then $v \in X_j$, for all $i \le j \le k$.
\end{definition}

The following result is central to our proof:

\begin{theorem}\textnormal{~\cite[Lemma~15 and Theorem~16]{berwanger2012dag}}
	\label{lemma:dagwidthEqv}
	A digraph G has DAG-width $k$ if and only if the cop player has a cop-monotone winning 
	strategy in the $k$-cops and robber game on G.
\end{theorem}

Therefore, in order to prove that DAG-width of a graph $G$ is at most $k$, it suffices to find a
cop-monotone winning strategy for the cop player in the $k$-cops and robber game on $G$. In the next section,
we present such a strategy and argue its correctness. We will later (in Section~\ref{sec:decompose}) give a 
second proof, not using the $k$-cops and robber game, of the DAG-width of control-flow graphs. As such, 
Section~\ref{sec:proof} is not required for our main result, but is a useful tool for gaining insight
into the structure of control-flow graphs,
and also provides a way of proving a lower bound on the DAG-width.

\section{Cops and Robbers on Control Flow Graphs}
\label{sec:proof}
Let $G=(V, E)$ be the control flow graph of a structured program $P$.
Recall that we characterize a loop element $L$ by its \emph{entry} and \emph{exit} points and refer to it by
the pair $(L^{\entry}, L^{\exit})$. 
%
We now present the following strategy $f$ for the cop player in the cops and robber game on $G$ with \emph{three} 
cops. 
\begin{enumerate}
	\item We will throughout the game maintain that at this point $X(1)$ occupies $L^{\entry}$, $X(2)$ occupies $L^{\exit}$, and $r \in \inside(L)$, for
some loop element $L$.

(In the first round $L \coloneqq L_\phi$, where $L_\phi$ is the hypothetical loop element that encloses $G$.  Regardless of the initial position of the robber, $r\in \inside(L_\phi)$. The cops $X(1)$ and
$X(2)$ are not used in the initial step.)

	\item Now we move the cops: 
		\begin{enumerate}
			\item If $r \in \belongs(L)$, move $X(3)$ to $r$. 
			\item Else, since $r\in \inside(L)$, we must have 
				$r \in \inside(L_i)$ for some loop $L_i$
				directly nested under $L$.
				Move $X(3)$ to $L_i^{\exit}$.
		\end{enumerate}
	
	\item Now the robber moves, say to $r'$. Note that $r' \in (\inside(L)\cup\{\texttt{stop}\})$ since $r\in \inside(L)$ and
		$X(1)$ and $X(2)$ block all paths from there to $(\outside(L) \setminus \{\texttt{stop}\})$.

	\item One of four cases is possible:
		\begin{enumerate}
			\item $r' = X(3)$. Then we have now caught the robber and we are done.
			\item $r' = \texttt{stop}$. Move $X(3)$ to \texttt{stop} and we will catch the robber in next move since the robber cannot 
				leave \texttt{stop}.
			\item $r' \in \inside(L_i)$, i.e., the robber stayed inside the same loop that it was before.  Go to step 5.
			\item $r' \in (\inside(L) \setminus \inside(L_i))$, i.e., the robber left the inside of the loop that it was in. Go back to step $2$.
		\end{enumerate}

	\item We reach this case only if the robber $r'$ is inside $L_i$,
	and $X(3)$ had moved to $L_i^{\exit}$ in the step before.  Thus
	cop $X(3)$ now blocks movements of $r'$ to 
	$(\outside(L_i) \setminus \{\texttt{stop}\})$.  We must do one
	more round before being able to recurse:
	\begin{enumerate}
		\item Move $X(1)$ to $L_i^{\entry}$.
		\item The robber moves, say to $r''$.
			By the above, $r'' \in (\inside(L_i)\cup \texttt{stop})$.
		\item If $r'' = L_i^{\entry}$, we have caught
			the robber.
		If $r'' = \texttt{stop}$, we can catch the
						robber in the next move.
		\item In all remaining cases, $r'' \in \inside(L_i)$. Go back to step $1$ with $L \coloneqq L_i$, $X(2)$ as $X(3)$ and
						$X(3)$ as $X(2)$.
				\end{enumerate}
\end{enumerate}

For a step-by-step annotated example, see Appendix~\ref{appendix:example}.
It should be intuitive that we make progress if we reach Step (5), since we
have moved to a loop that is nested more deeply.  It is much less obvious
why we make progress if we reach 4(a).  To prove this, we introduce the
notion of a distance function $\dist(v,L^{\exit})$, which measures roughly 
the length of the longest path from $v$ to $L^{\exit}$, except that we do
not count vertices that are inside loops nested under $L$.  Formally:
\begin{definition}
Let $L$ be a loop element of $G$ and $v \in \inside(L)$. Define 
$\dist(v, L^{\exit}) = \max_{P}(|P\cap\belongs(L)|),$
where $P$ is a directed simple path from $v$ to $L^{\exit}$ that uses only vertices
in $\inside(L)$ and does not use $L^{\entry}$.
\end{definition}

\begin{lemma}
\label{lem:move}
When the robber moves from $r$ to $r'$ in step (3), then
$\dist(r',L^{\exit})\leq \dist(r,L^{\exit})$.
The inequality is strict 
if $r\in \belongs(L)$ and $r'\neq r$.
\end{lemma}
\begin{proof}
Let $P$ be the directed path from $r$ to $r'$ along which the robber moves.
Notice that $L^{\entry}\not\in P$ since $X(1)$ is on $L^{\entry}$.
Let $P'$ be the path that achieves the maximum in $\dist(r',L^{\exit})$;
by definition $P'$ does not contain $L^{\entry}$.

$P\cup P'$ may contain directed cycles, but if $C$ is such a cycle
then no vertices of $C$ are in $\belongs(L)$ by 
Corollary~\ref{cor:cycle}.
So if we let $P_s$ be what remains of $P\cup P'$ after removing all
directed cycles then $P_s\cap \belongs(L) = (P\cup P') \cap \belongs(L)$.
Since $P_s$ is a simple directed path from $r$ to $L^{\exit}$ that
does not use $L^{\entry}$, therefore
$\dist(r,L^{\exit}) \geq |P_s\cap \belongs(L)|\geq 
|P' \cap \belongs(L)| = \dist(r',L^{\exit})$ as desired.
If $r'\neq r$, then $P'$ cannot possibly
include $r$ while $P_s$ does, and so  if additionally
$r\in \belongs(L)$ then the inequality is strict.
\end{proof}

\begin{lemma}
	\label{lemma:winningStrategy}
	The strategy $f$ is winning.
\end{lemma}
\begin{proof}
Clearly the claim holds if the robber ever moves to \texttt{stop}, so assume
this is not the case.
Recall that at all times the  strategy maintains a loop $L$ such
that two of the cops are at $L^{\entry}$ and $L^{\exit}$.  
We do an induction on the number of loops that are nested in $L$.

So assume first that no loops are nested inside $L$.  Then $\inside(L)
=\belongs(L)$, and by Lemma~\ref{lem:move} the distance of the robber
to $L^{\exit}$ steadily decreases since $X(3)$ always moves onto the robber,
forcing it to relocate. Eventually the robber must get caught.

For the induction step, assume that there are loops nested inside
$L$.  If we ever reach step (5) in the strategy, then the enclosing loop
$L$ is changed to $L_i$, which is inside $L$ and hence has fewer loops inside
and we are done by induction.  But we must reach step (5) eventually (or
catch the robber directly), because with every execution of (3) the robber
gets closer to $L^{\exit}$:
\begin{itemize}
\item If $r\in \belongs(L)$, then this follows directly from Lemma~\ref{lem:move}
	since $X(3)$ moves onto $r$ and forces it to move.
\item If $r\in \inside(L_i)$, and we did not reach step (5), then $r$
	must have left $L_i$ using $L_i^{\exit}$.  Notice that
	$dist(r_i,L^{\exit})=dist(L_i^{\exit},L^{\exit})$ due to our
	choice of distance-function.  Also notice that
	$L_i^{\exit} \in \belongs(L)$ since $L_i$ was directly nested
	under $L$.  We can hence view the robber as having moved to
	$L_i^{\exit}$ (which keeps the distance the same) and then to
	the new position (which strictly decreases the distance by 
	Lemma~\ref{lem:move} to $L_i^{\exit}$).
\qed
\end{itemize}
\end{proof}

\begin{lemma}
	\label{lemma:monotoneStrategy}
	The strategy $f$ is cop-monotone.
\end{lemma}
\begin{proof}
	We must show that the cops do not re-visit a previously visited vertex at any step of the
	strategy $f$. We note that since \texttt{stop} is a sink in $G$ and the cops move to \texttt{stop}
	only if the robber was already there, it will never be visited again. Now the only steps which we need
	to verify are (2) and (5a).

Observe that while we continue in step (2), the cops $X(1)$ and $X(2)$ always stay 
at $L^{\entry}$ and $L^{\exit}$ respectively, and $X(3)$ {\em always} stays
at a vertex in $\belongs(L)$.   (This holds because $L_i$ was chosen to be
nested directly under $L$ in Case (2b), so $L_i^{\exit}\in \belongs(L)$.)   
Also notice that
$\dist(X(3), L^{\exit})=\dist(r, L^{\exit})$ for as long as we stay in step (2),
because vertices in $\inside(L_i)$ do not count towards the distance.
In the previous proof we saw that
the distance of the robber to $L^{\exit}$ strictly decreases while we
continue in step (2).  So
$\dist(X(3), L^{\exit})$ also strictly decreases while we stay
in step (2), and so $X(3)$ never re-visits a vertex.

During step (5), the cops move to $L_i^{\entry}$ and $L_i^{\exit}$ and from 
then on will only be at vertices in $\inside(L_i)\cup \{L_i^{\exit}\}$.
Previously cops were only in $\belongs(L)$ or in $\outside(L)$.   These
two sets intersect only in $L_i^{\exit}$, which is occupied throughout
the transition by $X(3)$ (later renamed to $X(2)$).  Hence no cop can
re-visit a vertex and the 
	strategy $f$ is cop-monotone.\qed
\end{proof}

With this, we have shown that the DAG-width is at most 3.  This is tight.
\begin{lemma}
	\label{lemma:counterExample}
The exists a control-flow graph that has DAG-width at least 3.
\end{lemma}
\begin{proof}
	By Theorem~\ref{lemma:dagwidthEqv}, it suffices to show that
	the robber player has a winning strategy against two cops.
	We use the graph from Fig.~\ref{fig:counterexample} and
	the following strategy:
	\begin{enumerate}
\item Start on vertex $5$.  We maintain the invariant that
	at this point the robber is at $5$ or $6$, and
	there is at most one cop on vertices $\{5,6,7,8\}$.
	This holds initially when the cops have not been placed yet.
\item If (after the next helicopter-landing) there will still be at most one
	cop in $\{5,6,7,8\}$, then move such that afterwards the robber
	is again at $5$ or $6$.  (Then return to (1).)
	The robber can always get to one of $\{5,6\}$ as follows:
	If no cop comes to where the robber is now,
	then stay stationary.  If one does, then get to the other
	position using cycle $5\rightarrow 6\rightarrow 7\rightarrow 5$;
	this cannot be blocked since one cop is moving to the robbers
	position and only one cop is in $\{5,6,7,8\}$ afterwards.
\item If (after the next helicopter-landing) both cops will be in
	$\{5,6,7,8\}$, then ``flee'' to vertex $9$ along the directed path 
	$\{5\mbox{ or }6\} \rightarrow 8 \rightarrow 1 \rightarrow 2 \rightarrow 9$.  

\item Repeat the above steps with positions $\{9,10\}$, cycle $\{9,10,11\}$ 
	and escape path $\{9\mbox{ or }10\}\rightarrow 12\rightarrow 1\rightarrow 2
	\rightarrow 5$ symmetrically.
	\end{enumerate}
	
Thus the robber can evade capture forever by toggling 
between the two loop elements $L_1$ and $L_2$ and hence has a winning strategy.\qed
\end{proof}

In summary:
\begin{theorem}
	\label{theorem:cfgDagwidth}
	The DAG-width of control-flow graphs is at most $3$ and this is tight for some control-flow graphs.
\end{theorem}

\begin{figure}[t!]
	\centering
	\captionsetup[sub float]{captionskip=5pt}
	\subfloat[A control flow graph $G$]{
		\includegraphics[scale=0.35]{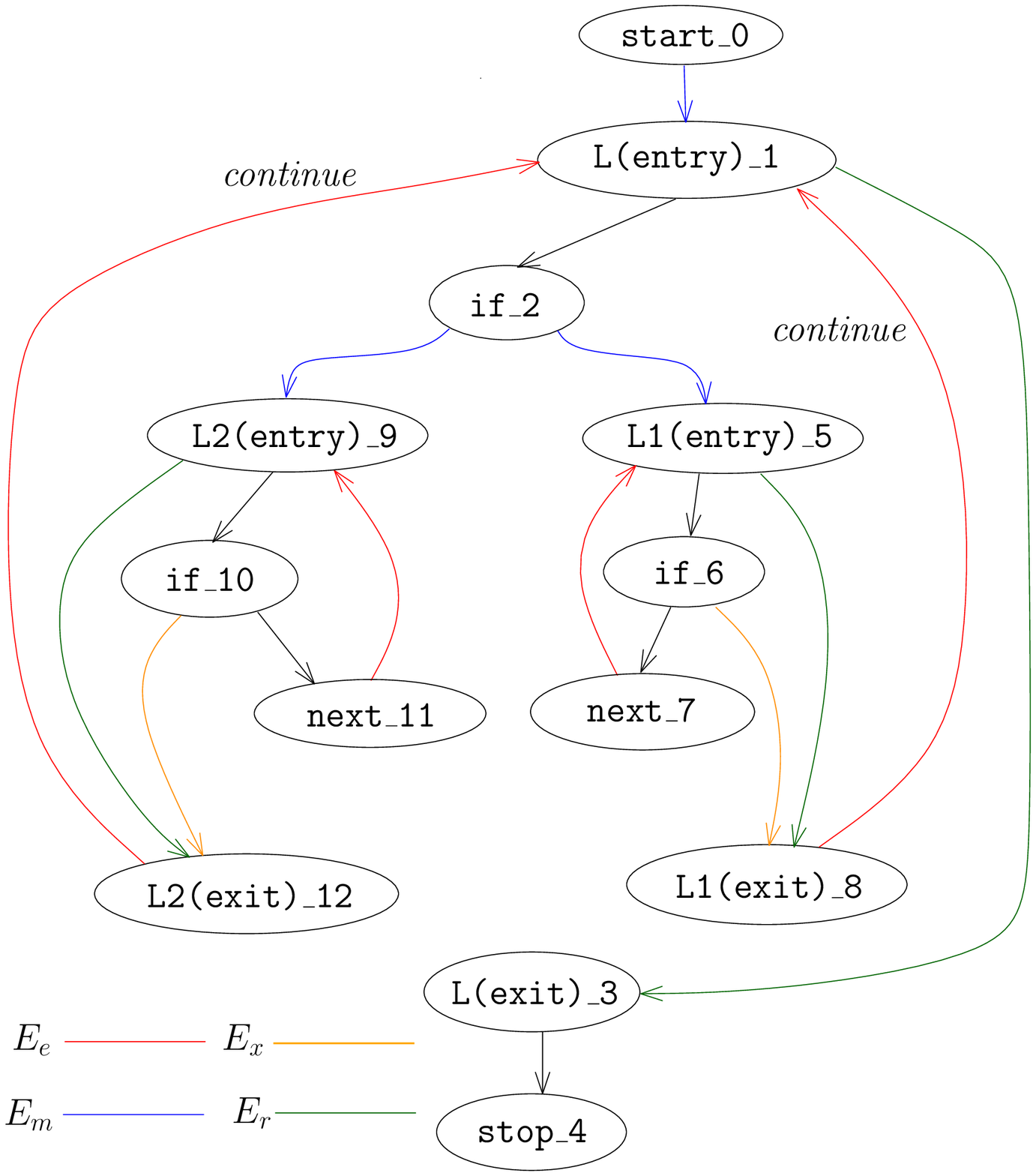}
		\label{fig:counterexample}
	}\hfill
	\subfloat[DAG Decomposition of G]{
		\includegraphics[scale=0.32]{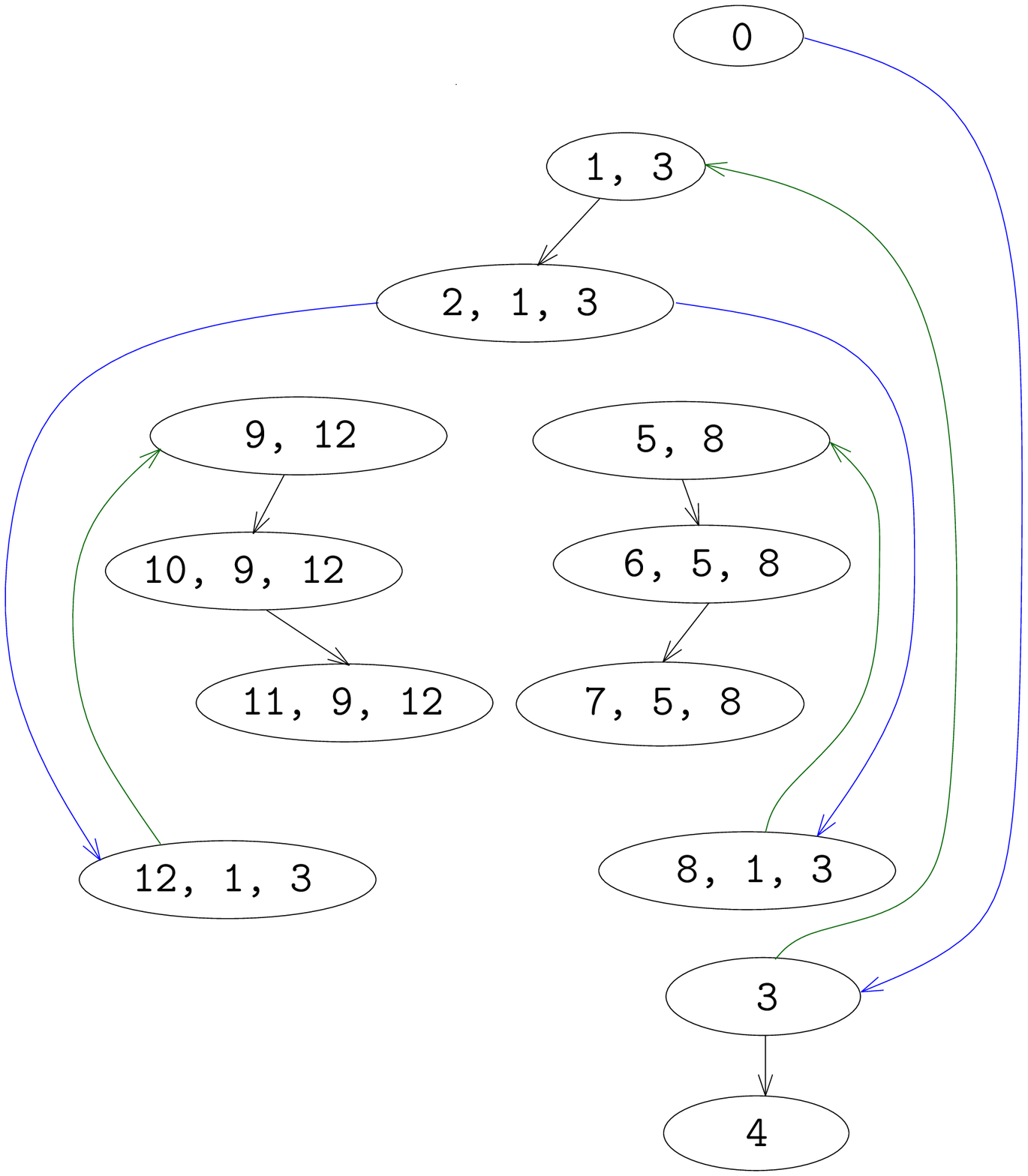}
		\label{fig:decomposition}
	}
	\caption{The robber player has a winning strategy on $G$ against two cops}
\end{figure}

\section{Computing the DAG decomposition }
\label{sec:decompose}

We already showed that control-flow graphs have DAG-width at most 3 (Theorem~\ref{theorem:cfgDagwidth}). In this section we
show how we can construct an associated DAG decomposition with few edges.
\subsection{DAG-width}
We first state precise definition of DAG-width, for which we need the following notation.
For a directed acyclic graph (DAG) $D$, 
use $u \preceq_D v$ to denote that there is a directed path 
from $u$ to $v$ in $D$.
\begin{definition} [DAG Decomposition] Let $G = (V, E)$ be a directed graph. A DAG decomposition of $G$ consists of a DAG $D$
	and an assignment of bags $X_i \subseteq V$ to every node $i$ of $D$ such that:
	\begin{enumerate}
		\item \emph{(Vertices covered)} $\bigcup X_i = V$.
		
		\item \emph{(Connectivity)} 
For any $i\preceq_D k \preceq_D j$ we have
			$X_i \cap X_j \subseteq  X_k$.
		
\item \emph{(Edges covered)} 
\begin{enumerate}
\item For any source $j$ in $D$, any $u \in X_j$, and any edge $(u, v)$ 
in $G$, there exists a successor bag $X_k$ of $X_j$ with $v \in X_k$.
\item For every arc $(i, j)$ in $D$, any $u \in (X_j \setminus X_i)$, 
and any edge $(u, v)$ in $G$, there exists a successor-bag $X_k$ of $X_j$
with $v \in X_k$.
\end{enumerate}
Here a \emph{successor-bag} of $X_j$ is a bag $X_k$ with $j \preceq_D k$.
	\end{enumerate}
\end{definition}

Note that for the ease of understanding we have rephrased the edge-covering condition of 
the original in~\cite{berwanger2012dag}. 
(Appendix~\ref{appendix:eqv} proves the equivalence.)

\subsection{Constructing a DAG decomposition}
While we already know that the DAG-width of control-flow graphs is at most
3, we do not know the DAG-decomposition and its number of edges $M$ that
is needed for the runtime for~\cite{fearnley2011time}.
There is a method to get the DAG decomposition from a winning strategy for the
$k$-cops and robber game~\cite{berwanger2012dag},
but it only shows $M \in O(n^{k+2})$. 
We now construct a DAG decomposition for control-flow graphs 
directly. Most importantly, it has $M \in O(n)$ edges,
thereby making~\cite{fearnley2011time} even more efficient for control-flow graphs.

Let $G = (V, E)$ be a control-flow graph. We present the following algorithm to construct a DAG decomposition $(D, X)$ of 
$G$.

\begin{Algorithm}[Construct the DAG]
	\label{algo:constructDAG}
	\begin{enumerate}
	\item Start with $D = G$. That is $V(D) = V$ and $E(D) = E$.
	
	\item  Remove all backward arcs.  Thus, 
	let $E_e$ be all backward edges of $G$;  
	recall that each of them connects from a node $v\in \belongs(L)$ 
		to $L^{\entry}$, for some loop element $L$.
	Remove all arcs corresponding to edges in $E_e$ from $D$.
		
	\item Remove all arcs leading to a loop-exit.  Thus, let $E_x$ be
	all edges $(u, v)$ in $G$ such that 
	$u \in \belongs(L)$ and $v = L^{\exit}$ for some loop element $L$.
	Recall that these arcs are attributed to \texttt{break} statements.
	Remove all arcs corresponding to edges in $E_x$ from $D$.
	
	\item Re-route all arcs leading to a loop-entry.  Thus, let $E_m$ be 
	all edges $(u, v)$ in $G$ such that 
	$u \in \outside(L) \setminus L^{\exit}$ and $v = L^{\entry}$ for some loop 
	element $L$.    For each such edge, remove the corresponding 
	edge in $D$ and replace it by an arc $(u,L^{\exit}$).  Let $A_m$ be
	those re-routed arcs.
	Note that now $\indeg_D(L^{\entry}) = 0$ since we also 
	removed all backward edges.

	\item Reverse all arcs surrounding a loop.  Thus, let
	$E_r$ be the edges in $G$ of the form $(L^{\entry}, L^{\exit})$ for some
	loop element $L$.  For each edge, reverse the corresponding
	arc in $D$.  Let $A_r$ be the resulting arcs.
	\end{enumerate}
\end{Algorithm}

Note that there is a bijective mapping $\mathcal{D}$ from the vertices 
in $G$ to the nodes in $D$.
For ease of representation, assume that $v_1, v_2, .., v_n$ are
the vertices of $G$ and $1, 2, .., n$ are the corresponding nodes in $D$. 
%
We now fill the bags $X_i$:

\begin{quotation}
For every  vertex $v_i \in V$, set $X_i \coloneqq \{v_i, L^{\entry}, L^{\exit}\}$,
				where $L$ is the loop element such that $v_i \in \belongs(L)$.
\end{quotation}

A sample decomposition is shown in Figure~\ref{fig:decomposition}.  Clearly the construction can be done in linear time and
digraph $D$ has $O(|E|)=O(n)$ edges.  One easily verifies the following:

\begin{observation}
	\label{obs:uniqueIntroduce}
For every arc $(i, j) \in E(D)$, $X_j \setminus X_i = v_j$.
\end{observation}

It remains to show that $(D, X)$ is a valid DAG decomposition.
\begin{enumerate}
	\item \textbf{$D$ is a DAG}. 
We claim that $G-E_e$ is acyclic.  For if it contained a directed cycle $C$,
then let $L$ be a loop element with $C\subseteq \inside(L)$, but $C\not\subseteq 
\inside(L_i)$ for any loop element $L_i$ nested under $L$.  Therefore $C$
contains a vertex of $\belongs(L)$. By Corollary~\ref{cor:cycle} then $L^{\entry}$
belongs to $C$, so $C$ contains a backward edge.  This is impossible since we
delete the backward edges.

Adding arcs $A_m$ cannot create a cycle since each arc in it is a shortcut
for the 2-edge path from outside $L$ to $L^{\entry}$ to $L^{\exit}$.
In $G-E_e-E_x$ there is no directed path from $L^{\entry}$ to $L^{\exit}$, since
such a path would reside inside $L$, and the last edge of it belongs to $E_x$.
In consequence adding arcs $A_r$ cannot create a cycle either.  Hence $D$ is acyclic.
\item \textbf{Vertices Covered}.  By definition each $v_i$ is contained in its bag $X_i$.
\item \textbf{Connectivity}.  Let $i\preceq_D k \preceq_D j$ be three nodes
	in $D$.  Recall that their three bags are
	$\{v_i,L_i^{\entry},L_i^{exit}\}$,
	$\{v_k,L_k^{\entry},L_k^{exit}\}$ and
	$\{v_j,L_j^{\entry},L_j^{exit}\}$, where $L_i,L_j,L_k$ are the
	loop elements to which $v_i,v_j,v_k$ belong.
	Nothing is to show unless $X_i\cap X_j\neq \emptyset$, which 
	severely restricts the possibilities:
\begin{enumerate}
\item 
Assume first that $L_i = L_j=L$.  Thus $v_i$ and $v_j$ belong to the
same loop element, and by the directed path between them, so does $k$.
So the claim holds since $X_i\cap X_j = \{L^{\entry},L^{\exit}\}$.

\item If $L_i \ne L_j$, then the intersection can be non-empty only if
$v_i = L_j^{\exit}$ (recall that $i\preceq_D j$).  But then
$X_i\cap X_j=\{L_j^{\exit}\}$, and the path from $i$ to $j$ must go
from $\mathcal{D}(L_j^{exit})$ to $\mathcal{D}(L_j^{\entry})$ to $j$.  
It follows that $v_k$ also belongs to $L_j$ and so $L_j^{\exit} \in X_k$
and the condition holds.
\end{enumerate}

	\item \textbf{Edges Covered}.
We only show the second condition; the first one is similar
	and easier since \texttt{start} is the only source.  Let $(i,j)$ be an arc in $D$.
By Observation~\ref{obs:uniqueIntroduce}, 
		$v_j$ is the only possible vertex in $X_j \setminus X_i$.
		Let $e = (v_j, v_l)$ be an edge of $G$. We have the following cases:
		\begin{enumerate}
			\item If $e \in (E_e \cup E_x \cup E_r)$, then $v_l \in \{L^{\entry}, L^{\exit}\}$ and $X_j = \{v_j, L^{\entry}, L^{\exit}\}$
				itself can serve as the required successor-bag.
			\item If $e \in E_m$, then $v_l = L^{\entry}, v_j \in \outside(L)$. 
				We re-routed $(v_j,L^{\entry})$ as arc $(j,\mathcal{D}(L^{\exit})$
				and later added an arc $(\mathcal{D}(L^{\exit}),\mathcal{D}(L^{\entry}))$,
				so $l$ is a successor of $j$ and $X_l$ can serve  as the
				required successor-bag.
			\item Finally, if $e \in E \setminus (E_e \cup E_x \cup E_r \cup E_m)$, 
				then $(j, l)$ is an arc in $D$ and $X_l$ is the required successor-bag.
		\end{enumerate}
\end{enumerate}

We conclude:

\begin{theorem}
	\label{theorem:result}
	Every control-flow graph $G = (V, E)$ has a DAG decomposition of 
width 3 with $O(|V|)$ vertices and edges.
	It can be found in linear time.
\end{theorem}

\comment{TB: I've thrown out all mentioning of ``nice''.  The proof you give
uses nothing but the face that we have $O(n)$ edges in the DAG-decomposition,
and hence is rather trivial.}


\section{Conclusion}
\label{sec:conclusion}
In this paper, 
we showed that control-flow graphs have DAG-width at most $3$.  Our proof comes
with a linear-time algorithm to find such a DAG decomposition, and it has
linear size.  Since algorithms that are tailored to small DAG-width are 
typically exponential in the DAG-width, this should
improve the run time of such algorithms for control-flow graphs. 
The specific application that motivated this paper was the
DAG-width based algorithm for parity games from~\cite{fearnley2011time};
using our DAG decomposition should turn this into  a more practical
algorithm for software model checking.
 (See Appendix~\ref{sec:parityGames} for more details).
The run-time is still rather slow for large $n$. One natural open problem
is hence to develop even faster algorithms for parity games on digraphs
that come from control flow graphs. Our 
simple DAG-decomposition that is directly derived from the control
flow graph might be helpful here.

Our result also opens directions for future research in other 
related application areas.  For example, can we use the small
DAG-width of control-flow graphs for  faster analysis of the 
worst-case execution time 
(which is essentially a variant of the longest-path problem)?

\InputIfFileExists{paper.bbl}

\newpage
\appendix

\section{Additional Examples}
\label{appendix:example}
\begin{Example}
\textit{Winning strategy $f$ from Section~\ref{sec:proof} applied on graph in Figure~\ref{fig:counterexample}.}\\
	
	We will refer to the vertices
by their indices. Suppose that the initial position of the robber is vertex $0$, and the robber plays a lazy 
strategy, that is, he stays where he is unless a cop comes there, otherwise, he moves to the closest cop-free vertex. 
Then, the following two sequences of positions are possible. Note that the labels on the transitions represent
the corresponding steps of the strategy $f$ whereas $X(i) = \phi$ indicates that the cop $X(i)$ was not used.

\begin{itemize}
	\item[I.] $\xrightarrow{~~1~~} (\{\phi, \phi, \phi \}, 0)$
		  $\xrightarrow{~2a~} (\{\phi, \phi, 0\}, 1)$
		  $\xrightarrow{~2b~} (\{\phi, \phi, 3\}, 1)$
		  $\xrightarrow{~5~} (\{1, \phi, 3\}, 2)$
		  $\xrightarrow{~2a~} (\{1, 3, 2\}, 9)$
		  $\xrightarrow{~2b~} (\{1, 3, 12\}, 9)$
		  $\xrightarrow{~5~} (\{9, 3, 12\}, 10)$
		  $\xrightarrow{~2a~} (\{9, 12, 10\}, 11)$
		  $\xrightarrow{~2a~} (\{9, 12, 11\}, 11)$
	  	  $\xrightarrow{~4a~}$ \emph{stop}

	  \item[II.] $\xrightarrow{~~1~~} (\{\phi, \phi, \phi \}, 0)$ 
		  $\xrightarrow{~2a~} (\{\phi, \phi, 0\}, 1)$
		  $\xrightarrow{~2b~} (\{\phi, \phi, 3\}, 1)$ 
		  $\xrightarrow{~5~} (\{1, \phi, 3\}, 2)$
          $\xrightarrow{~2a~} (\{1, 3, 2\}, 5)$
          $\xrightarrow{~2b~} (\{1, 3, 8\}, 5)$ 
		  $\xrightarrow{~5~} (\{5, 3, 8\}, 6)$
		  $\xrightarrow{~2a~} (\{5, 8, 6\}, 7)$ 
		  $\xrightarrow{~~2a~~}$
		  $(\{5, 8, 7\}, 7) \xrightarrow{~4a~}$ \emph{stop}

\end{itemize}
\end{Example}

\section{Properties of DAG-width}
\label{appendix:dagWidthProperties}

	\label{appendix:eqv}
We aim to show that our edge-covering condition is equivalent to the one given by Berwanger et al.
\cite{berwanger2012dag}.  We first review their concepts.

\begin{definition}[Guarding] 
Let $G = (V, E)$ be a digraph and $W, V^\prime \subseteq V$. We say that $W$ \emph{guards} $V^\prime$ 
if, for all $(u,v) \in E$, if $u \in V^\prime$ then $v \in V^\prime \cup W$.
\end{definition}

The original edge-covering condition was the following:
\begin{quotation}
(D3) For all edges $(d,d')\in E(D)$, $X_d\cap X_{d'}$ guards $X_{\succeq d'}\setminus X_d$, where
	$X_{\succeq d'}$ stands for $\bigcup_{d'\preceq_D d''} X_{d''}$.
For any source $d$, $X_{\succeq d}$ is guarded by $\emptyset$.
\end{quotation}
For easier comparison we re-state here our edge-covering condition:
\begin{quotation}
(3a) For any edge $(i,j)$ in $E(D)$, any vertex $u\in X_j\setminus X_i$, and any edge $(u,v)$ in $G$,
there exists a successor-bag $X_k$ of $X_j$ that contains $v$.

(3b) For source $j$ in $D$, any vertex $u\in X_j$, and any edge $(u,v)$ in $G$,
there exists a successor-bag $X_k$ of $X_j$ that contains $v$.
\end{quotation}
We will only show that the first half of (D3) is equivalent to (3a); one can
similarly show that the second half of (D3) is equivalent to (3b).
We first re-phrase (D3) partially by switching to our notation, and partially
by inserting the definition of guarding; clearly (D3') is equivalent to the
first half of (D3).
\begin{quotation}
(D3') For any edge $(i,j)$ in $E(D)$, any vertex $u\in X_{\succeq j}\setminus X_i$
and any edge $(u,v)$ in $G$, we have $v\in X_{\succeq j} \cup (X_i \cap X_j)$.
\end{quotation}
But $X_i\cap X_j \subseteq X_j \subseteq X_{\succeq j}$, so we can immediately
simplify this again to the following equivalent:
\begin{quotation}
(D3'') For any edge $(i,j)$ in $E(D)$, any vertex $u\in X_{\succeq j}\setminus X_i$
and any edge $(u,v)$ in $G$, we have $v\in X_{\succeq j}$.
\end{quotation}
At the other end, we can also simplify (3a), since we now have the shortcut
$X_{\succeq j}$ for vertices in a successor-bag of $X_j$.
\begin{quotation}
(3a') For any edge $(i,j)$ in $E(D)$, any vertex $u\in X_j\setminus X_i$, and any edge $(u,v)$ in $G$,
we have $v\in X_{\succeq j}$.
\end{quotation}
Thus (D3'') and (3a') state nearly the same thing, except that for
(D3'') the claim must hold for significantly more vertices $u$.  As such,
(D3'')$\Rightarrow$(3a') is trivial since $X_j\setminus X_i\subseteq X_{\succeq j}\setminus X_i$.

For the other direction, we need to work a little harder.  Assume (3a') holds.
To show (D3''), fix one such choice
of edge $(i,j)$ in $E(D)$ and $(u,v)$ in $E(G)$ with $u\in X_{\succeq j}\setminus X_i$.
We show that $v\in X_{\succeq j}$ using induction one the number of successors of
$j$ in $D$.  If there are none, then $X_{\succeq j}=X_j$ and (D3'') holds since (3a') does.
Likewise (D3'') holds if $u\in X_j\setminus X_i$ since (3a') holds.  This leaves the case where
$u\in (X_{\succeq j}\setminus X_j)$.  Thus $u$ belongs to some strict successor bag of $X_j$,
and hence there exists an arc $(j,k)$ with $u\in (X_{\succeq k}\setminus X_i)$.  
Node $k$ has
fewer successors than $j$, and so by induction (D3'') holds for edge $(j,k)$.
We know $u\in (X_{\succeq k}\setminus X_i)$ and $u\not\in X_j\setminus X_i$, so
$u\in (X_{\succeq k}\setminus X_j)$.  So applying (D3'') we know $v\in X_{\succeq k}
\subseteq X_{\succeq j}$ and hence (D3'') also holds for edge $(i,j)$.

\newpage
\section{Application to Software Model Checking}
\label{sec:parityGames}

In this section we will discuss how exactly the DAG-width based algorithm from~\cite{fearnley2011time}
for solving parity games is used for software model checking 
(which is essentially the $\mu$-calculus model checking problem on control-flow graphs).
We start with briefly discussing the modal $\mu$-calculus, parity games
and how to convert the $\mu$-calculus model checking problem 
to the problem of finding a winner in parity games. 

\subsection{Parity Game}
A \emph{parity game} $\mathcal{G}$ consists of a directed graph $G = (V, E)$ called
\emph{game graph} and a parity function $\lambda: V \rightarrow \mathbb{N}$ (called
\emph{priority}) that assigns
a natural number to every vertex of $G$.

The game $\mathcal{G}$ is played between two players $P_0$ and $P_1$ who move
a shared token along the edges of the graph $G$. The vertices $V_0 \subset V$ 
and $V_1 = V \setminus V_0$ are assumed to be owned by $P_0$ and $P_1$ respectively.
If the token is currently on a vertex in $V_i$ (for $i=0,1$), then player $P_i$ 
gets to move the token, and moves it to a successor of his choice.
This results in a possibly infinite sequence of
vertices called \emph{play}. If the play is finite, the player who is unable to move
loses the game. If the play is infinite, $P_0$ wins the game if the largest occurring
priority is even, otherwise $P_1$ wins.

A solution for the parity game $\mathcal{G}$ is a partitioning of $V$ into $V_0^{w}$ and
$V_1^{w}$, which are respectively the vertices from which $P_0$ and $P_1$ have a 
winning strategy. Clearly, $V_0^{w}$ and $V_1^{w}$ should be disjoint.

\subsection{Modal $\mu$-calculus}
The modal $\mu$-calculus (see~\cite{Stirling:2001} for a good introduction) is a fixed-point logic comprising a set of formulas defined by the
following syntax:
\begin{align*}
	\phi ::= X \mid \phi_1 \wedge \phi_2 \mid \phi_1 \vee \phi_2 \mid [\cdot]\phi \mid \langle \cdot \rangle \phi \mid \nu X .\phi \mid \mu X .\phi
\end{align*}

Here, $X$ is the set of propositional variables, and $\nu, \mu$ are maximal and minimal fixed-point operators respectively.
The \emph{alternation depth} of a formula is the number of syntactic alternations between the maximal fixed-point operator,
$\nu$, and the minimal fixed-point operator, $\mu$.

Given a formula $\phi$, we say that a $\mu$-calculus formula $\psi$ is a subformula of $\phi$, if we can obtain $\psi$ from $\phi$
by recursively decomposing as per the above syntax. For example, the formula $\nu X (P \wedge X)$ has four subformulas: $\nu X (P \wedge X)$,
$(P \wedge X)$, $P$ and $X$. The \emph{size} of a formula is the number of its subformulas.

\subsection{$\mu$-calculus Model Checking to Parity Games}

A model $M = (S, T)$ is represented as a digraph with the set of states $S$ as vertices
and the transitions $T$ as edges.  The {\em $\mu$-calculus model checking problem}
consists of testing whether  a given modal-$\mu$-calculus formula $\phi$ applies to $M$.
As mentioned earlier, given $M$ and $\phi$, there exists a way to construct a
parity game  instance $\mathcal{G} = (G, \lambda)$ such that $\phi$ applies if and
only if the parity game can be solved.
See e.g.~\cite{Stirling:2001} \footnote{Alternatively, see pages 20-23: Obdr\v{z}\'{a}lek, Jan. 
Algorithmic analysis of parity games. PhD thesis, University of Edinburgh, 2006}
We note the following relevant points of this transformation:

	\begin{enumerate} [ref={\theObservation.\arabic*}]
	\item Let $Sub(\phi)$ be the set of all subformulas of $\phi$ and $m = |Sub(\phi)|$ be the size of the formula $\phi$. 
		For every $\psi \in Sub(\phi)$ and $s \in S$ we create a vertex $(s, \psi)$ in $G$.
		Therefore, $|V(G)| = m \cdot |S|$.
	
	\item \label{obs:possibleEdges}
		For every $s \in S$, let $V_s \subseteq V(G)$ be the set of vertices of $G$ 
		$\{(s,\psi): \psi \in Sub(\phi)\}$.  Clearly, $|V_s| = m$.
		It holds that for any $s, t \in S$, there is an edge between any two vertices 
		$u \in V_s$ and $v \in V_t$ of $G$ only if $(s, t) \in T$.
		
	\item The number of priorities $d$ in $\mathcal{G}$ is equal to the alternation depth of the formula $\phi$ plus two. That is,
		for a $\mu$-calculus formula with no fixed-point operators, the number of priorities is at least $2$.
	\end{enumerate}

\subsection{$\mu$-calculus Model Checking on Control Flow Graphs}

Recall that given a $\mu$-calculus formula of length $m$ and a control-flow graph $G = (V, E)$, we can create a 
parity game graph $G^\prime$ with $m \cdot |V|$ vertices. Now, we can use either of the treewidth or DAG-width
based algorithms from~\cite{fearnley2011time} for solving the parity game on $G^\prime$. We discuss them individually.

\paragraph{Treewidth based algorithm}
Recall that this runs in $O(|V| \cdot (k+1)^{k+5} \cdot (d+1)^{3k+5})$ time where $k$ is the treewidth of the 
game graph $G^\prime$. Using Thorup's result~\cite{thorup1998all} we can obtain a tree decomposition $(\mathcal{T}, \mathcal{X})$
of $G$ with width at most $6$. This means that each bag of the tree decomposition
contains at most $7$ vertices.
Using Observation~\ref{obs:possibleEdges}, we can now obtain a tree decomposition
$(\mathcal{T}^\prime, \mathcal{X}^\prime)$ for $G^\prime$ 
from $(\mathcal{T}, \mathcal{X})$ by replacing every $s \in \mathcal{X}_i$ by $V_s$, for all $\mathcal{X}_i \in \mathcal{X}$.
Note that the width of $\mathcal{T}^\prime$ will be $7 \cdot m-1$. 

\paragraph{DAG-width based algorithm}
This runs in $O(|V| \cdot M \cdot k^{k+2} \cdot (d+1)^{3k+2})$ time where $k$ is the DAG-width of the
game graph $G^\prime$ and $M$ is the number of edges in the DAG decomposition. Using our main result 
(Theorem~\ref{theorem:result}), we can obtain a DAG decomposition $(D, X)$ of width $3$ and $M \in O(|V|)$.
As in the previous case, we can obtain a DAG decomposition of $G^\prime$ from $(D,X)$ by replacing every 
$s \in X_i$ with $V_s$, for all $X_i \in X$. Note that this will have width $3 \cdot m$ and $M \in O(|V|)$.

\medskip
We can see that even for the smallest possible values $m = 1$ and $d = 2$, the treewidth based algorithm runs in
$O(|V| \cdot 7^{11} \cdot 3^{23}) = O(|V| \cdot 10^{20})$ time. For the same values,
the DAG-width based algorithm runs in $O(|V|^2 \cdot 3^{5} \cdot 3^{11}) = O(|V|^2 \cdot 10^{7})$,
which is better unless $|V|\geq 10^{13}$.  Of course the actual run-times may 
be influenced  by the constants hidden behind the asymptotic notations, but it
is fair to assume that the DAG-width based algorithm will be faster for most
practical scenarios, especially as $m$ and $d$ increase.

\end{document}